\documentclass[twoside,11pt,leqno]{article}
\usepackage{amsthm,amsmath,amssymb,epsf}
\usepackage[english]{babel}

\newcommand{\xx}{{\bf x}}

\newcommand{\R}{{\mathbb R}}
\newcommand{\N}{{\mathbb N}}
\newcommand{\EE}{{\mathbb E}}
\newcommand{\PP}{{\mathbb P}}

\newcommand{\1}{{\bf 1}}

\newtheorem{defn}{Definition}

\newtheorem{thm}{Theorem}
\newtheorem{prop}{Proposition}

\renewenvironment{proof}{\noindent{\bf Proof:} }{\hfill $\square$ \\}

\begin{document}
\bibliographystyle{plain}

\thispagestyle{empty}
\begin{center}
{\Large \sc
A $J$--function for inhomogeneous point processes}\\[.5in]

\noindent
{\large M.N.M. van Lieshout}\\[.2in]
\noindent
{\em CWI\/} \& {\em Eindhoven University of Technology\/} \\
\mbox{} \\
P.O. Box 94079, 1090 GB Amsterdam, The Netherlands \\[.2in]
\end{center}

\begin{verse}
{\footnotesize
\noindent
{\bf Abstract}\\
\noindent
We propose new summary statistics for intensity-reweighted moment 
stationary point processes that generalise the well known $J$-, empty 
space, and nearest-neighbour distance distribution functions, represent
them in terms of generating functionals and conditional intensities, and 
relate them to the inhomogeneous reduced second moment function. Extensions 
to space time and marked point processes are briefly discussed.\\[0.2in]

\noindent
{\em Keywords \& Phrases:}
conditional intensity, empty space function, generating functional, $J$-function, 
nearest-neighbour distance distribution function, inhomogeneity, intensity-re\-weight\-ed 
moment stationarity, marked point process, minus sampling estimator, product 
density, reduced second moment measure, spatial interaction.

\noindent
{\em 2000 Mathematics Subject Classification:}
60D05, 60G55, 62M30.
}
\end{verse}
\section{Introduction}

The analysis of data in the form of a map of (marked) points often starts 
with the computation of summary statistics. Some statistics are based on 
inter-point distances, others on the average number of points in sample 
regions, or geometric information. For a survey of the state of the art
and pointers to the literature, the reader is referred to the recent handbook 
of spatial statistics \cite{Gelf10}.

In the exploratory stage, it is usually assumed that the data constitute a 
realisation of a stationary point process and deviations from a homogeneous 
Poisson process are studied to suggest a suitable model. Although stationarity
is a convenient assumption, especially if -- as is often the case -- only a 
single map is available, in many areas of application, though,  heterogeneity 
{\em is\/} present. To account for possible non-stationarity, Baddeley {\em 
et al.\/} \cite{BaddetalK00} defined a reduced second moment function by 
considering the random measure obtained from the mapped point pattern by 
weighting each observed point according to the (estimated) intensity at its 
location.  Gabriel and Diggle \cite{GabrDigg09} took this idea further into 
the domain of space time point processes. 

In this paper, our aim is to define an extension of the $J$-function 
\cite{LiesBadd96} that is able to accommodate spatial and/or temporal
inhomogeneity. The idea underpinning the $J$-function is to compare the point 
pattern around a typical point in the map to that around an arbitrarily chosen 
origin in space in order to gain insight in the interaction structure of the point 
process that generated the data. The power of the $J$-function in hypothesis 
testing was assessed in \cite{Chen03} and \cite{ThonLies99}. Extensions to 
multivariate point processes were proposed by \cite{LiesBadd99}, and window 
based $J$-functions suggested by \cite{BaddetalJ00} and \cite{Chen03}. For 
applications in agriculture, astronomy, forestry and geology, see 
\cite{FoxBadd02,Kers98,Kersetal98,Kersetal99,Paul02,Steietal01}.

The plan of this paper is as follows. In Section~\ref{S:prelim} we fix
notation and recall some basic concepts from stochastic geometry. In
Section~\ref{S:stats} we describe the most important summary statistics
that are being used in exploratory analysis of point patterns under the 
assumption of stationarity. Section~\ref{S:J} introduces the new statistic
$J_{\rm{inhom}}$ and gives representations of it in terms of generating 
functionals and conditional intensities. Section~\ref{S:theo} is devoted to 
the explicit computation of $J_{\rm{inhom}}$ for some important classes 
of point process models. In Section~\ref{S:estim} we develop a minus 
sampling estimator and apply it in Section~\ref{S:examples} to simulated 
examples. The paper closes with suggestions for further extensions to 
space time and marked point processes.

\section{Preliminaries and notation}
\label{S:prelim}
Throughout this paper, let $X$ be a simple point process on $\R^d$.
Its {\em intensity measure\/} $\Lambda$ is defined by
\[
  \Lambda(B) = \EE \left[ \sum_{x \in X} 1\{ x \in B \} \right]
\]
for Borel sets $B\subseteq \R^d$. We assume that $\Lambda$ is locally 
finite, i.e.\ $\Lambda(B) < \infty$ whenever $B$ is bounded, and absolutely 
continuous with respect to Lebesgue measure so that
\[
  \Lambda(B) = \int_B \lambda(x) \, dx
\]
for some non-negative measurable function $\lambda$ referred to as
{\em intensity function\/}. Heuristically speaking $\lambda(x) \, dx$ is 
the probability of observing some point in the infinitesimal region $dx$ and
represents the heterogeneity of $X$.

Note that the intensity measure is also known as the {\em first order 
factorial moment measure\/} of $X$. Higher order factorial moment measures 
$\Lambda^{(n)}$, $n\in\N$, are defined by
\[
  \Lambda^{(n)}(B_1 \times \cdots \times B_n) = \EE \left[
  {\sum}^{\neq}_{x_1, \dots, x_n \in X} 1\{ x_1 \in B_1; \dots; x_n\in B_n \}
  \right],
\]
where the superscript $\mbox{}^{\neq}$ indicates that the sum is taken over
all $n$-tuples of distinct points and the $B_i$ are Borel subsets of $\R^d$. 
As the intensity measure, $\Lambda^{(n)}$ is not necessarily locally finite, nor
guaranteed to have a Radon--Nikodym derivative. If $\Lambda^{(n)}$ is 
absolutely continuous with respect to the $n$-fold product of Lebesgue measures,
\[
  \Lambda^{(n)}(B_1 \times \cdots \times B_n) = \int_{B_1} \cdots \int_{B_n} 
  \rho^{(n)}(x_1, \dots, x_k) \, dx_1 \cdots dx_k
\]
for some non-negative measurable function $\rho^{(n)}$ called 
{\em $n$-th order product density\/} of $X$. Note that $\rho^{(n)}$ is
permutation invariant and satisfies the integral equation
\[
  \EE \left[ {\sum}^{\neq}_{x_1, \dots, x_n \in X} g(x_1, \dots, x_n) \right]
=
  \int \cdots \int g(x_1, \dots, x_n) \, \rho^{(n)}(x_1, \dots, x_n) \,
  dx_1 \cdots dx_n
\]
for all non-negative, measurable functions $g \geq 0$. Thus, $\rho^{(n)}(dx_1,
\dots, dx_n)$ may be interpreted as the infinitesimal probability of finding
points of $X$ at each of $dx_1, \dots, dx_n$. For further details, see for
example \cite{Gelf10,Illietal08,Lies00}. 

In the physics literature, {\em $n$-point correlation functions} tend to
be used instead of product densities \cite{Peeb80}.  They are defined
recursively by
\begin{eqnarray*}
\nonumber
\xi_1 & \equiv & 1; \\
\frac{\rho^{(n)}(x_1, \dots, x_n)}{\lambda(x_1) \cdots \lambda(x_n)} & = &
 \sum_{k=1}^n \sum_{D_1, \dots, D_k}
 \xi_{n(D_1)}(\xx_{D_1}) \cdots \xi_{n(D_k)}(\xx_{D_k}),
\end{eqnarray*}
where the last sum ranges over all partitions $\{ D_1, \dots, D_k \}$ of
$\{ 1, \dots, n \}$ in $k$ non-empty, disjoint sets, and
the $\xx_{D_j} = \{ x_i : i \in D_j \}$, $j=1, \dots, k$ form the 
corresponding partition of points.  Since for a Poisson point process 
$\xi_n \equiv 0$ for $n>1$, heuristically speaking $n$-point correlation 
functions account for the excess due to $n$-tuples in comparison to a Poisson 
point process with the same intensity function.

\section{Summary statistics}
\label{S:stats}

Summary statistics are used by spatial statisticians as tools for exploratory
data analysis, testing, and model validation purposes. Popular examples include
the {\em nearest neighbour distance distribution function\/} $G$, the {\em
empty space function\/} $F$, the {\em reduced second moment function\/}
$K$ and the $J$-function. More specifically, for a stationary point process
$X$ with intensity $\lambda > 0$,  
\begin{equation}
\label{e:stats}
\left\{
\begin{array}{lll}
F(t) & = & \PP( X \cap B(0,t) \neq \emptyset ), \\
G(t) & = & \PP^{!0}( X \cap B(0,t) \neq \emptyset ), \\
K(t) & = & \EE^{!0} \left[ 
            \sum_{x\in X} 1\{ x \in B(0,t) \} \right] / \lambda, \\
J(t) & = & \left(1-G(t)\right) / \left( 1-F(t) \right),
\end{array} \right.
\end{equation}
where $B(0,t)$ is the closed ball of radius $t\geq 0$ centred at the origin and
$\PP^{!0}$ denotes the reduced Palm distribution of $X$. For further details
about these and other summary statistics, see for example \cite{Illietal08}.
Note that the $J$-function is defined only for $t$ such that $F(t) < 1$. 
Values larger than one indicate inhibition, whereas $J(t) < 1$ suggests 
clustering, but note the caveats against drawing too strong conclusions 
in \cite{BedfBerg97}.

All statistics defined in (\ref{e:stats}) can be expressed in terms of product 
densities. Indeed, if the second order factorial moment measure exists as a locally
finite measure with Radon--Nikodym derivative $\rho^{(2)}(x_1, x_2)
= \rho^{(2)}(||x_1 - x_2||)$,
\[
  K(t) = \int_{B(0,t)} \frac{\rho^{(2)}(||x||)}{\lambda^2} \, dx
       = \int_{B(0,t)} \left( 1 + \xi_2(||x||) \right) dx.
\]
Clearly, $K$ depends only on product densities up to order two.
In contrast, the empty space function depends on product densities
of all orders \cite{Whit79}, 
\[
  F(t) = - \sum_{n=1}^\infty \frac{ (-1)^n }{ n! }
  \int_{B(0,t)} \cdots \int_{B(0,t)} \rho^{(n)}( x_1, \dots, x_n ) \,
  dx_1 \cdots dx_n 
\]
provided all order product densities exist and the series is absolutely 
convergent, that is, 
$\limsup_{n\to\infty} \left(
  \frac{\Lambda^{(n)}(B(0,t)^n)}{n!} \right)^{1/n} < 1$.
Similarly,
\[
  G(t) = - \sum_{n=1}^\infty \frac{(-1)^n}{n!} 
    \int_{B(0,t)} \cdots \int_{B(0,t)}
    \frac{\rho^{(n+1)}( 0, x_1, \dots, x_n )}{\lambda}\, dx_1 \cdots dx_n,
\]
provided that the series is absolutely convergent. Thence \cite{Lies06},
\[
  J(t) =  1 +  \sum_{n=1}^\infty \frac{(-\lambda)^n}{n!} J_{n}(t)
\]
for all $t \geq 0$ for which $F(t)<1$, where $J_{n}(t) =
\int_{B(0,t)} \cdots \int_{B(0,t)}
\xi_{n+1}( 0, x_1, \dots, x_n) \, dx_1 \cdots dx_n$.
If product densities of all orders do not exist, one may truncate the series.
Indeed, using only product densities up to second order gives
\[
  J(t) - 1 \approx - \lambda \left( K(t) - | B(0,t) | \right)
\]
so the $K$-function can be seen as a second order approximation to the 
$J$-function.

For non-stationary point processes, the definitions in (\ref{e:stats}) depend
on the choice of origin and adaptations are called for. To this end, Baddeley 
{\em et al.\/} \cite{BaddetalK00} introduced the notion of {\em second 
order intensity-reweighted stationarity\/}. A point process $X$ possesses 
this property if the random measure
\[
   \Xi = \sum_{x\in X} \frac{ \delta_x }{ \lambda(x) }
\]
is second-order stationary. Here, $\delta_x$ denotes the Dirac measure
that places a single point at $x$. Clearly if $\Xi$ is stationary, it is
also second-order stationary but the converse does not hold. Examples
of second order intensity-reweighted stationary point processes 
include Poisson point processes, the random thinning of a stationary
point process, and log Gaussian Cox processes driven by a Gaussian 
random field with a translation invariant covariance function.
Cluster processes, as well as more general superposition processes,
typically are not second order intensity-reweighted stationary.

For a second order intensity-reweighted stationary point process, 
an inhomogeneous K-function \cite{BaddetalK00} can be defined by
\[
  K_{\rm{inhom}}(t) := \frac{1}{|B|} \EE \left[ {\sum}^{\neq}_{x, y\in X}
     \frac{1_B(x) \, 1\left\{ y \in {B(x,t)} \right\} }{ 
     \lambda(x) \, \lambda(y) } 
  \right]
\]
regardless of the choice of bounded Borel set $B\subset\R^d$ and using
the convention $a/0=0$ for $a\geq 0$. Indeed, 
\(
K_{\rm{inhom}}(t) = {\cal{K}}_{\Xi}(B(0,t)\setminus\{ 0 \}),
\)
where ${\cal{K}}_{\Xi}$ is the reduced second moment measure of
the random measure $\Xi$.

Gabriel and Diggle \cite{GabrDigg09} restrict themselves to point
processes $X$ that are simple and have locally finite moment measures
of first and second order. Additionally they assume that $X$ has an 
intensity function $\lambda$ that is bounded away from zero 
and a pair correlation function 
\[
g(x,y) = g(|| x - y ||) = \frac{\rho^{(2)}(x,y)}{\lambda(x)\, \lambda(y)}
\]
that depends only on $|| x - y ||$.
Clearly in this case 
\[
K_{\rm{inhom}}(t) = \frac{1}{|B|}  \int_B \int_{B(0,t)} g(||z||) \, dz \, dx
= \int_{B(0,t)} g(||z||) \, dz,
\]
which for any inhomogeneous planar Poisson process reduces to $\pi t^2$.

Baddeley {\em et al.\/} \cite{BaddetalK00} briefly discuss how to define
empty space and nearest neighbour distance distribution functions for 
inhomogeneous point processes. First, for given $x\in\R^d$ and $t\geq 0$,
they propose to determine $r(x,t)$ by solving
\[
t = \int_{B(x, r(x,t))} \lambda(y) \, dy,
\]
then set
\begin{eqnarray*}
F_x(t) & = & \PP( d(x, X) \leq r(x,t) ) \\
G_x(t) & = & \PP^{!x}( d(x, X) \leq r(x,t) ),
\end{eqnarray*}
where $d(x,X)$ denotes the shortest distance from $x$ to a point of $X$.
For Poisson processes, the above definitions do not depend on $x$ and are 
both equal to $1-e^{-t}$. The obvious drawback of such an approach is that 
$r(x,t)$ may be hard to compute in practice.  Moreover, the definitions 
depend on $x$ as well as $t$. Our goal in the present paper is to give an 
alternative definition of $F$, $G$, and $J$ for intensity-reweighted
moment stationary point processes based on their representation in 
terms of product densities that does not depend on the choice of origin
and is easy to use in practice.

\section{Inhomogeneous $J$-function}
\label{S:J}

Let $X$ be a simple point process on $\R^d$  whose intensity function 
$\lambda$ exists and is bounded away from zero with $\inf_x 
\lambda(x) = \bar\lambda >0$. Assume that for all $n\in\N$ the 
$n^{\rm{th}}$ order factorial moment measure exists as a locally finite measure
and has a Radon--Nikodym derivative $\rho^{(n)}$ with respect to the 
$n$-fold product of Lebesgue measure $\ell$ with itself for which the 
corresponding $n$-point correlation function $\xi_n$ is translation 
invariant, that is, 
$\xi_n(x_1 + a, \dots, x_n + a) = \xi_n(x_1, \dots, x_n)$ 
for almost all $a\in \R^d$. We shall call such a point process {\em
intensity-reweighted moment stationary\/}. Note that a fortiori $X$ is
second order intensity-reweighted stationary. Moreover, a
stationary point process is also intensity-reweighted moment stationary.

\begin{defn}
\label{d:J}
Let $X$ be an intensity-reweighted moment stationary point process. Set
\[
J_n(t) = \int_{B(0,t)} \cdots \int_{B(0,t)}
\xi_{n+1}( 0, x_1, \dots, x_n) \, dx_1 \cdots dx_n
\]
and define
\[
  J_{\rm{inhom}}(t) = 1 + 
  \sum_{n=1}^\infty \frac{(-\bar \lambda)^n}{n!} J_n(t),
\]
for all $t\geq 0$ for which the series is absolutely convergent, that is, for which
$\limsup_{n\to\infty}$ $ \left( \frac{\bar\lambda^n}{n!} |J_n(t)| \right)^{1/n}
 < 1$.
\end{defn}

A few special cases deserve to be mentioned. For a Poisson point process 
with intensity function $\lambda(\cdot)$, as the $n$-point correlation functions 
vanish for $n>1$, so do the $J_n$ whence $J_{\rm{inhom}}(t) \equiv 1$ for 
all $t\geq 0$. Furthermore, if $X$ is stationary, $\bar\lambda = \lambda$ 
and by \cite[Prop.~4.2]{Lies06}, $J_{\rm{inhom}} \equiv J$. 

Like in the stationary case considered in Section~\ref{S:stats}, the series in
Definition~\ref{d:J} may be truncated, for example when $X$ is only 
second order intensity-reweighted stationary or not all $n$-point correlation 
functions exist. For $n=1$, we obtain
\[
J_{\rm{inhom}}(t) - 1 \approx - \bar\lambda \int_{B(0,t)} \xi_2(0,x) \, dx
 = - \bar \lambda \left( K_{\rm{inhom}}(t) - | B(0,t) | \right).
\]

In the remainder of this section, we rewrite $J_{\rm{inhom}}$ in terms of
generating functionals and conditional intensities. Recall that for any function
$v:\R^d \to [0,1]$ that is measurable and identically $1$ except on some
bounded subset of $\R^d$, the generating functional at $v$ is defined as
\[
G(v) = \EE \left[ \prod_{x\in X} v(x) \right],
\]
where by convention an empty product is taken to be $1$. The distribution
of $X$ is determined uniquely by its generating functional 
\cite[Prop.~7.4.II]{DaleVere88}. The factorial moment measures, provided they 
exist as locally finite measures, can be derived from the generating 
functional using its Taylor expansion \cite[Prop.\ 7.4.III]{DaleVere88}. Conversely, 
if product densities of all orders exist, let $u$ be a measurable function
with values in $[0,1]$ that has bounded support. Then 
\[
G(v:= 1 - u) = 1 + \sum_{n=1}^\infty \frac{(-1)^n}{n!} 
\int \cdots \int u(x_1) \cdots u(x_n) \, \rho^{(n)}(x_1, \dots, x_n) \,
dx_1 \cdots dx_n,
\] 
provided the series converges \cite[p.~109]{Stoyetal87}.

\begin{thm}
\label{t:gfl}
Write, for $t \geq 0$ and $a\in \R^d$,
\[
u_t^a(x) = \frac{ \bar\lambda \, 1\{ x \in B(a,t) \} }{ \lambda(x) } 
\]
and assume that 
\(
\limsup_{n\to\infty} \left( 
\frac{\bar\lambda^n}{n!}  \int_{B(0,t)}\cdots \int_{B(0,t)} 
\frac{\rho^{(n)}(x_1, \dots, x_n)}{\lambda(x_1) \cdots \lambda(x_n)} 
dx_1 \cdots dx_n \right)^{1/n} < 1 .
\)
Under the assumptions of Definition~\ref{d:J}, for almost all $a\in \R^d$,
\[
J_{\rm{inhom}}(t) = \frac{ G^{!a}\left( 1 - u_t^a \right) }
{ G\left( 1 - u_t^0 \right) }
\]
for all $t\geq 0$ for which the denominator is non-zero, where $G^{!a}$ is 
the generating functional of the reduced Palm distribution $\PP^{!a}$ at $a$,
$G$ that of $\PP$ itself. 
\end{thm}

Note that for a stationary point process, $u_t^a(x) = 1\{ x\in B(a,t) \}$, hence
\[
G\left( 1 - u_t^a \right) = \PP( X \cap B(a,t) = \emptyset )
= 1 - F(t).
\]
Therefore, the generating functional in the denominator can be interpreted as the 
inhomogeneous counterpart of the empty space function.  A similar interpretation 
holds for the numerator in terms of the nearest neighbour distance distribution 
function and one retrieves the classic definition of the $J$-function given in
Section~\ref{S:stats}. At this point it should be emphasised that the numerator 
and denominator in the definiton of $J_{\rm{inhom}}$ generalise respectively the 
nearest neighbour distance distribution function and empty space function.

\mbox{}

\begin{proof}
We begin by showing that 
\[
\EE^{!x}\left[ {\sum}^{\neq}_{x_1, \dots, x_n\in X} \prod_{i=1}^n
\frac{ 1\{ x_i \in B(x,t) \} }{\lambda(x_i) }  \right]
=
\int_{B(0,t)} \cdots \int_{B(0,t)} \frac{\rho^{(n+1)}(0, x_1, \dots x_n)}
  {\lambda(0) \; \lambda(x_1) \cdots \lambda(x_n) } \; dx_1 \dots dx_n
\]
for almost all $x\in\R^d$. To see this, consider the functions
\[
g_A(x, X) = \frac{1\{x\in A\}}{\lambda(x)} \;
{\sum}^{\neq}_{x_1, \dots, x_n\in X}
\prod_{i=1}^n \frac{ 1\{ x_i \in B(x,t) \} }{\lambda(x_i) }  
\]
defined for all bounded Borel sets $A \subset \R^d$. By the
definition of Palm distributions and the Campbell--Mecke formula,
\[
\int \int g(x, \varphi) \; \lambda(x) \; d\PP^{!x}(\varphi) \; dx =
\EE\left[
\sum_{x\in X} g(x, X\setminus\{ x \} \right] .
\]
Using Fubini's theorem, for our choice of $g$ the left hand side 
can be written as
\[
 \int_A \EE^{!x}  \left[ 
{\sum}^{\neq}_{x_1, \dots, x_n\in X} \prod_{i=1}^n
\frac{ 1\{ x_i \in B(x,t) \} }{\lambda(x_i) }  \right] \; dx
\]
while the right hand side is equal to
\[
\EE\left[ 
{\sum}^{\neq}_{x, x_1, \dots, x_n} \frac{ 1\{ x\in A \} }{\lambda(x) }
 \; 
\prod_{i=1}^n \frac{ 1\{ x_i \in B(x,t) \} }{\lambda(x_i) }  
\right] .
\]
The expectation can be computed in terms of $\rho^{(n+1)}$ and equals
\[
\int_A \int_{B(x,t)} \cdots \int_{B(x,t)} \frac{\rho^{(n+1)}(x, x_1, \dots x_n)}
  {\lambda(x) \; \lambda(x_1) \cdots \lambda(x_n) } \; dx\; dx_1 \cdots dx_n 
= 
\]
\[
\int_A \int_{B(0,t)} \cdots \int_{B(0,t)} \frac{\rho^{(n+1)}(0, x_1, \dots x_n)}
  {\lambda(0) \; \lambda(x_1) \cdots \lambda(x_n) } \; dx\; dx_1 \cdots dx_n
\]
by the translation invariance of the $n$-point correlation functions. Hence
\[
\EE^{!x}\left[ {\sum}^{\neq}_{x_1, \dots, x_n\in X} \prod_{i=1}^n
\frac{ 1\{ x_i \in B(x,t) \} }{\lambda(x_i) }  \right]
\]
is constant for almost all $x\in\R^d$.

Next, note that
\[
\prod_{x\in X} \left( 1 - \frac{ \bar \lambda \, 1\{ x \in B(a,t) \} 
  }{\lambda(x)} \right)
= 
1 + \sum_{n=1}^\infty \frac{(- \bar \lambda)^n}{n!} 
{\sum}^{\neq}_{x_1, \dots, x_n \in X} \prod_{i=1}^n
\frac{1\{ x_i \in B(a,t)\} }{\lambda(x_i)}.
\]
Since the number of points in $X \cap B(a,t)$ is almost surely finite,
the expressions are well-defined under the convention that an empty 
product takes the value one.  Consequently, for almost all $a$,
\begin{equation}
\label{e:Gpalm}
 G^{!a}( 1 - u_t^a )   =  
 1 + \sum_{n=1}^\infty \frac{(-\bar \lambda)^n}{n!}
\int_{B(0,t)} \cdots
\int_{B(0,t)} \frac{ \rho^{(n+1)}(0, x_1, \dots, x_n) }{
\lambda(0) \; \lambda(x_1) \cdots \lambda(x_n)} \, dx_1 \cdots dx_n
\end{equation}
provided the power series in the right hand side is absolutely convergent.

By the discussion preceeding the statement of the theorem,
\begin{equation}
\label{e:G}
 G( 1 - u_t^0 )   =  
 1 + \sum_{n=1}^\infty \frac{(-\bar \lambda)^n}{n!}
\int_{B(0,t)} \cdots
\int_{B(0,t)} \frac{ \rho^{(n)}(x_1, \dots, x_n) }{
\lambda(x_1) \cdots \lambda(x_n)} \, dx_1 \cdots dx_n,
\end{equation}
since the power series in the right hand side is assumed to be absolutely convergent.

Upon recalling the definition of the $n$-point correlation functions and 
splitting into terms that do or do not contain the origin, one obtains that the 
right hand side of (\ref{e:Gpalm}) is equal to
\[
 1 + \sum_{n=1}^\infty \frac{(-\bar \lambda)^n}{n!}
\sum_{D\subseteq \{ 1, \dots, n \}} J_{n(D)}(t) \sum_{k=1}^{n-n(D)}
\sum_{ \stackrel{D_1, \dots, D_k \neq \emptyset \mbox{ disjoint}}{\cup D_j = 
\{1, \dots, n\} \setminus D} } I_{n(D_1)} \cdots I_{n(D_k)} 
\]
(with $\sum_{k=1}^0 = 1$) which in turn can be written as
\begin{equation}
\label{e:gfl}
\left[ 1 + \sum_{n=1}^\infty \frac{(-\bar \lambda)^n}{n!}
 J_{n}(t) \right]
\times 
\left[ 1 + \sum_{m=1}^\infty \frac{(-\bar \lambda)^m}{m!} \sum_{k=1}^m
\sum_{ \stackrel{D_1, \dots, D_k \neq \emptyset \mbox{ disjoint}}{\cup D_j = \{1, \dots, m\}}}
 I_{n(D_1)} \cdots I_{n(D_k)} \right]
\end{equation}
where 
\[
I_n = \int_{B(0,t)} \cdots \int_{B(0,t)}
      \xi_n(x_1, \dots, x_n) \,
      dx_1 \cdots dx_n
\]
and $n(D)$ denotes the cardinality of the set $D$.  The sum 
over $k$ in the rightmost term of (\ref{e:gfl}) can be written as
\[
\int_{B(0,t)} \cdots \int_{B(0,t)} \frac{\rho^{(m)}(x_1, \dots, x_m)
}{\lambda(x_1) \cdots \lambda(x_m)} \, dx_1 \cdots dx_m,
\]
hence the second term in (\ref{e:gfl}) is equal to the right hand side of (\ref{e:G}).
Finally, since both sums in (\ref{e:gfl}) are absolutely convergent, so is 
(\ref{e:Gpalm}), an observation that completes the proof.
\end{proof}


Next, we focus our attention on {\em conditional intensities\/} $\lambda(x;X)$, 
$x\in\R^d$. Assuming they exist, they are defined in integral terms by 
\[
\EE\left[ \sum_{x\in X} g(x, X\setminus \{ x \} ) \right] =
\int \EE^{!x}\left[ g(x, X) \right]  \lambda(x) \,dx = 
\int \EE\left[ g(x, X) \lambda(x; X) \right] dx
\]
for any non-negative measurable function $g$.

\begin{thm}
\label{t:ci}
Assume that $X$ admits a conditional intensity and define the random
variable $W_a(X) := \prod_{x\in X} \left( 1 - u_t^a(x) \right)$. Then, under 
the assumptions of Theorem~\ref{t:gfl}, $ \EE \left[ W_a(X) \right] = 0$ implies
\(
\EE\left[ \lambda(a; X) \, W_a(X) / \lambda(a) \right] = 0
\)
for almost all $a\in \R^d$, and otherwise
\[
J_{\rm{inhom}}(t) = 
\EE\left[ \frac{\lambda(a; X)}{\lambda(a)} \,  W_a(X) \right]
 / \, \EE W_a(X),
\]
the $W_a$-weighted expectation of $\lambda(a;X)/\lambda(a)$.
\end{thm}

Consequently, $J_{\rm{inhom}}(t) \leq 1 \Leftrightarrow 
\rm{Cov}\left( \frac{ \lambda(a; X)}{\lambda(a)}, W_a(X) \right) \leq 0$ with
a similar statement for the opposite inequality sign. 

\mbox{}

\begin{proof}
Consider the functions
\[
g_A(x, X) = \frac{1\{ x \in A\}}{\lambda(x)} \,  \prod_{y\in X} 
\left( 1 - \frac{\bar \lambda 1\{ y\in B(x,t) \}}{\lambda(y)} \right)
\]
defined for all bounded Borel sets $A\subset \R^d$.  Arguing as in the proof
of Theorem~\ref{t:gfl} and using the definition of conditional intensities, one
obtains
\[
\int_A \EE^{!x} \left[ \prod_{y\in X} 
\left( 1 - \frac{\bar\lambda\, 1\{ y\in B(x,t) \}}{\lambda(y)} \right) \right] 
dx =
\int_A \EE\left[ \frac{\lambda(x; X)}{\lambda(x)} \,
\prod_{y\in X} \left( 1 - \frac{\bar\lambda\, 1\{ y\in B(x,t) \}}{
\lambda(y)} \right) \right] dx.
\]
Hence,
\[
\EE^{!x} \left[  \prod_{y\in X}
\left( 1 - \frac{\bar \lambda \, 1\{ y\in B(x,t) \}}{\lambda(y)} \right) \right]
= 
\EE\left[ \frac{\lambda(x; X)}{\lambda(x)} \,  \prod_{y\in X}
\left( 1 - \frac{\bar \lambda \, 1\{ y\in B(x,t) \}}{\lambda(y)} \right) \right]
\]
for almost all $x\in\R^d$. An appeal to Theorem~\ref{t:gfl} completes the proof.
\end{proof}

\section{Theoretical examples}
\label{S:theo}
\subsection{Poisson process}

Let $X$ be a Poisson point process with intensity function
$\lambda: \R^d \to \R^+$ that is bounded away from zero. Since
$\rho^{(n)}(x_1, \dots, x_n) = \prod_i \lambda(x_i)$, the
$n$-point correlation functions vanish for $n>1$, so $J_{\rm{inhom}}
(t) \equiv 1$ for all $t\geq 0$.  

The generating functional of $X$ is 
\[
G(1-u) = \exp\left[ - \int u(x) \, \lambda(x) \, dx \right].
\]
In particular, for the function $u=u_t^0$ defined in Theorem~\ref{t:gfl},
$G(1 - u_t^0) = \exp\left[ - \bar \lambda | B(0,t) | \right]$.
Also, since according to Slivnyak's theorem for a Poisson point process
$\PP^{!0} = \PP$, $G^{!0}(1 - u_t^0) = G(1-u_t^0)$. Finally, the conditional
intensity $\lambda(\cdot, X)$ of $X$ coincides almost everywhere with 
the intensity function $\lambda(\cdot)$.

\subsection{Location dependent thinning}

Let $X$ be a simple, stationary point process on $\R^d$ for which product
densities $\rho^{(n)}$ of all orders exist. Let $p : \R^d \to (0,1)$ be a 
measurable function that is bounded away from zero and consider the thinning 
of $X$ with retention probability $p(x)$ as in Example 8.2 of \cite{DaleVere88}. 
Since the process is simple,  the product densities $\rho^{(n)}_{\rm{th}}$ of the 
thinned point process can be expressed in terms of those of $X$ by
\(
 \rho^{(n)}_{\rm{th}}(x_1, \dots, x_n) = 
 \rho^{(n)}(x_1, \dots, x_n) \prod_{i=1}^n p(x_i).
\)
In particular, the intensity function of the thinned point
process is
\(
 \lambda_{\rm{th}}(x) = \lambda \, p(x),
\)
where $\lambda>0$ is the intensity of $X$.  Consequently, 
\[
\frac{\rho^{(n)}_{\rm{th}}(x_1, \dots, x_n)}{\lambda_{\rm{th}}(x_1) 
\cdots \lambda_{\rm{th}}(x_n) } = \frac{\rho^{(n)}(x_1, \dots, x_n)}{\lambda^n}.
\]
Therefore, the $n$-point correlation functions of the thinned point process
coincide with those of the underlying stationary point process $X$,
$\xi_n^{\rm{th}}(x_1, \dots, x_n) = \xi_n(x_1, \dots, x_n)$, and inherit 
the property of translation invariance. Hence
\(
J_n^{\rm{th}}(t) 
\)
is equal to the $J_n$-function of the underlying point process $X$.
As the intensity function of the thinned point process is bounded from
below by $\lambda \bar p$ where $\bar p$ is the infimum
of the retention probabilities,
\[
J_{\rm{inhom}}^{\rm{th}}(t) = 1 + 
  \sum_{n=1}^\infty \frac{(- \lambda \, \bar p)^n}{n!} J_n(t)
\]
for all $t\geq 0$ for which the series converges. Note that the power 
series coefficients are identical to those in the power series expansion of 
the $J$-function of $X$.

The generating functional of the thinned point process is 
\(
G_{\rm{th}}(v) = G( vp + 1 - p),
\)
where $G$ is the generating functional of $X$. Hence 
\[
G_{\rm{th}}\left(1 - \frac{\bar p}{p(\cdot)} \, 1\{ \cdot \in B(0,t) \} \right)  = 
G( 1 - \bar p \, 1\{ \cdot \in B(0,t) \} )  = 
\EE \left[ (1- \bar p)^{n(X\cap B(0,t))} \right],
\]
the generating function of the number of points of $X$ that fall in $B(0,t)$ 
evaluated at $1-\bar p$.

As the reduced Palm distribution of the thinned point process coincides with
a random location dependent thinning of the reduced Palm distribution of $X$
with retention probabilities given by the function $p$,
\[
G_{\rm{th}}^{!0}( 1 - u_t^0) = \EE^{!0} \left[ (1 - \bar p)^{n(X\cap B(0,t))} \right],
\]
so that under the assumptions of Theorem~\ref{t:gfl}
\[
J_{\rm{inhom}}^{\rm{th}}(t) = \frac{\EE^{!0} \left[ (1 - \bar p)^{n(X\cap B(0,t))} \right]}{
\EE \left[ (1 - \bar p)^{n(X\cap B(0,t))} \right]}.
\]

To conclude this example, note that the assumption of stationarity of the underlying 
point process $X$ may be weakened to intensity-reweighted moment stationarity.


\subsection{Scaling}
\label{S:scaling}

Let $X$ be a simple point process on $\R^d$ for which product densities 
$\rho^{(n)}$ of all orders exist. Let $c>0$ be a scalar constant and map the 
point pattern $X$ to $cX$. Then all order product densities $\rho^{(n)}_{cX}$
of $cX$ exist and are given by
\(
\rho^{(n)}_{cX}(x_1, \dots, x_n) = c^{-dn} \rho^{(n)}( x_1 / c, \dots, x_n/c ).
\)
In particular for $n=1$, $\lambda_{cX}(x) = c^{-d} \lambda(x/c)$. Therefore the
$n$-point correlation functions 
\[
\xi_n^{cX}(x_1, \dots, x_n) = \xi_n(x_1/c, \dots, x_n/c)
\]
of $cX$ are invariant under translations if and only if the $n$-point 
correlation functions $\xi_n$  of $X$ are, in which case the $J_n$-functions 
$J_n^{cX}$  of $cX$ are scaled versions 
\(
J_n^{cX}(t) = c^{dn} J_n(t/c)
\)
of the corresponding functions of $X$.
Furthermore, $\inf_{x\in\R^d} \lambda_{cX}(x) = \bar\lambda \, c^{-d}$,
so the inhomogeneous $J$-function of $cX$ is 
\[
J_{\rm{inhom}}^{cX}(t) = 1 + 
  \sum_{n=1}^\infty \frac{(- \bar \lambda \, c^{-d})^n}{n!} c^{dn} J_n(t/c) 
= 1 + \sum_{n=1}^\infty \frac{(- \bar \lambda )^n}{n!} J_n(t/c) = 
J_{\rm{inhom}}(t/c),
\]
the inhomogeneous $J$-function of $X$ evaluated at $t/c$ provided the series 
converges. Note that in contrast to the thinning case, the power series coefficients 
are not identical to those of the underlying point process $X$.

The generating functional of the scaled process is given by $G_{cX}(v) =
G(v(c \cdot))$, where $G$ is the generating functional of $X$, whence
\[
G_{cX}\left( 1 - \frac{c^{-d} \bar\lambda}{c^{-d} \lambda( \cdot /c)} 
  1\{ \cdot \in B(0,t) \} \right) = 
G\left( 1 - \frac{\bar\lambda}{ \lambda(\cdot)} 
  1\{ \cdot \in B(0,t/c) \} \right).
\]
Similarly, noting that
$d\PP^{!x}_{cX}(\varphi) = d\PP^{!x/c}(\varphi/c)$,
\[
G_{cX}^{!a}\left( 1 - \frac{c^{-d} \bar\lambda}{c^{-d} \lambda( \cdot /c)} 
  1\{ \cdot \in B(a,t) \} \right) = 
G^{!a/c}\left( 1 - \frac{\bar\lambda}{ \lambda(\cdot)} 
  1\{ \cdot \in B(a/c,t/c) \} \right).
\]
To conclude this example, a conditional intensity of $cX$ is obtained by
scaling that of $X$, i.e.\ $\lambda_{cX}(x, \varphi) = 
c^{-d} \lambda(x/c, \varphi/c)$ \cite{Hahnetal03}, 
from which we retrieve the formula 
$J_{\rm{inhom}}^{cX}(t) = J_{\rm{inhom}}(t/c)$ under the assumptions
of Theorem~\ref{t:ci}.

\subsection{Log Gaussian Cox process}

Write $Q$ for the distribution of a random measure defined in terms
of its Radon--Nikodym derivative $\Lambda$ with respect to Lebesgue
measure.  We assume that all moment measures of the random measure
exist and are locally finite. Let $X$ be the Cox process 
directed by the random intensity process $\Lambda$, that is, 
given a realisation $\Lambda=\lambda$, $X$ is a Poisson point process 
with intensity function $\lambda$. It follows from \cite[p.\ 262]{DaleVere88} 
that the factorial moment measures of $X$ exist and are equal to the moment 
measures of the driving random measure. Hence $X$ has product densities
\(
\rho^{(n)}(x_1, \dots, x_n) = \EE \left[ \prod_{i=1}^n \Lambda(x_i) \right].
\)
Moreover, the reduced Palm distribution of $X$ at $x$ is the distribution
of a Cox process with driving random measure distributed as $Q^x$, the
Palm distribution of the driving measure of $X$ at $x$ \cite[p.\ 141]{Stoyetal87}.

The class of log-Gaussian Cox processes \cite{Molletal98}
is especially convenient. For models in this class,
\[
\Lambda(x) = \exp\left[ Z(x) \right]
\]
where $Z$ is a Gaussian field. Such a field is defined fully by its mean and 
covariance function. Write $\mu(x)$ for the mean function, $\sigma^2(x)$ 
for the variance of $Z(x)$ and $r(x,y)$ for the correlation function. In other 
words, the covariance function of $Z$ is given by 
$\sigma(x) \, \sigma(y) \, r(x,y)$. Conditions have to be imposed on
these functions in order to make the resulting Cox process well-defined. 
In particular, the intensity function must be integrable almost surely, and 
\(
\Psi_\Lambda(B) = \int_B \Lambda(x) \, dx
\)
a finite random variable for all bounded Borel sets $B\subset \R^d$.
Moreover, the distribution of the random measure $\Psi_\Lambda$ must be
uniquely determined by that of $Z$.  Sufficient conditions are given in 
\cite[Thm.~3.4.1]{Adle81} for zero mean Gaussian processes. Therefore,
we additionally assume that the mean function $\mu$ is continuous and
bounded. Now, since
\(
\EE\left[ \prod_i \Lambda(x_i) \right] = 
\EE\left[ e^{\sum_i Z(x_i)} \right],
\)
the moment generating function of the normally distributed random 
variable $\sum_i Z(x_i)$ evaluated at $1$, 
\[
\rho^{(n)}(x_1, \dots, x_n) =
\exp\left[ \sum_{i=1}^n \left( \mu(x_i) + \frac{\sigma^2(x_i) }{2} \right) 
   + \sum_{i<j} \sigma(x_i) \sigma(x_j) r(x_i, x_j)
\right].
\]
Specialising to $n=1$, it follows that the intensity function is 
$\log \lambda(x) = \mu(x) + \sigma^2(x)/2$
whence
\[
\frac{\rho^{(n)}(x_1, \dots, x_n)}{\lambda(x_1) \cdots \lambda(x_n)} =
\exp\left[ \sum_{i<j} \sigma(x_i) \, \sigma(x_j) \, r(x_i, x_j)
\right] .
\]
Thus, if $\sigma(\cdot) \equiv \sigma> 0$ and $r(x,y) = r(x-y)$, $X$ is 
intensity-reweighted moment stationary and the intensity function
is bounded away from zero with infimum 
$\exp\left[ \sigma^2/2 + \inf_{x\in\R^d}\mu(x) \right]$.


In order to derive an explicit formula for $J_{\rm{inhom}}$, we turn to the 
generating functional. Recall that a Cox process has a generating functional
\cite[Prop.~8.5.1]{DaleVere88} defined by
$G(v) = \EE_Q \exp\left[ - \int (1 - v(x) ) \, \Lambda(x) \, dx \right]$.
Therefore, for the log-Gaussian Cox process
\[
G\left(1 - u_t^0\right) 
 =  \EE_Z \exp\left[ - \bar \mu \int_{B(0,t)} e^{Z(x) - \mu(x)} dx \right].
\]
where $\bar \mu$ denotes $\inf_{x\in\R^d} e^{\mu(x)}$. 

The Palm distributions $Q^x$ of a log Gaussian random measure are 
$\Lambda(x)=e^{Z(x)}$-weighted. To see this, note that the Campbell measure
evaluated at a bounded Borel set $B\subset\R^d$ and $F$ in the smallest 
$\sigma$-algebra for which $\Psi_\Lambda(B) = \int_B \Lambda(x) \, dx$ 
is finite for all such $B$, can be calculated as
\[
C(B\times F) =  \EE_Q\left[ 1_F( \Psi_\Lambda) \, \Psi_\Lambda(B) \right] 
=
 \int_B \lambda(x) \left[ 
   \int 1_F( \Psi_\Lambda) \; \frac{\Lambda(x)}{\lambda(x)} \; 
  dQ(\Psi_\Lambda) \right] dx
\]
by Fubini and the existence of a $\sigma$-finite intensity measure $\lambda(\cdot)$
that is bounded away from zero. 
Therefore,
\[
G^{!a}\left(1- u_t^a \right)  =   \EE _Z \left[ \frac{e^{Z(a) - \mu(a)}}{e^{\sigma^2/2}}
  \exp\left[ - \bar \mu \int_{B(a,t)} e^{Z(y) - \mu(y)} dy \right] 
\right].
\]
Since $Y(x) := Z(x) - \mu(x)$, $x\in\R^d$, is a stationary Gaussian process, 
the above generating functional does not depend on the choice of $a$.  Therefore, 
under the assumptions of Theorem~\ref{t:gfl},
\[
J_{\rm{inhom}}(t) = \frac{ \EE_Y e^{Y(0)} \left[ \exp\left[ - \bar\mu 
   \int_{B(0,t)} e^{Y(x)} dx \right] \right] }{
\EE_Y\left[ e^{Y(0} \right] \EE_Y \exp\left[ - \bar\mu
  \int_{B(0,t)} e^{Y(x)} dx \right] }.
\]
The mixed Poisson process considered in \cite{LiesBadd96}  is a special case.

Note that $J_{\rm{inhom}}(t) < 1$ if and only if the random variables $e^{Y(0)}$ and 
$e^{-\bar \mu \int_{B(0,t)} e^Y}$ are negatively correlated. The geostatistical models 
used in practice, for example the one we shall use in Section~\ref{S:examples},
have a positive, continuously decreasing, correlation function. Therefore, by
Pitt's theorem \cite{Pitt82},  the Gaussian fields defined by such correlation
functions are associated. Under the further conditions of \cite[Thm. 3.4.1.]{Adle81}, 
the sample functions $Y(\cdot)$ are almost surely continuous and hence the integral 
of $e^Y$ over $B(0,t)$ is uniquely defined and the limit of Riemann sums over ever 
finer partitions of $B(0,t)$.  Since $Y$ is associated, 
$\mbox{Cov}\left( e^{Y(0)}, e^{-c_i \sum_i e^{Y(x_i)}} \right) \leq 0$ 
for all finite sums with positive scalar multipliers $c_i > 0$. Upon taking the limit,
it follows that $J_{\rm{inhom}}(t) \leq 1$.

\section{Estimation}
\label{S:estim}

The goal of this section is to develop an estimator for the inhomogeneous
$J$-function of Definition~\ref{d:J}. For this purpose, we shall use the 
representation in terms of generating functionals of Theorem~\ref{t:gfl}
and apply the minus sampling principle outlined in \cite[p.~127]{Stoyetal87}.

Specifically, let Let $W\subset\R^d$ be a compact set with non-empty 
interior and suppose the point process $X$ is observed in $W$. For clarity
of exposition, we assume that the intensity function $\lambda$ is known. If 
it is not, it can be estimated (for instance using kernel estimation 
\cite{BermDigg89}) and plugged into the estimators outlined below.

In order to estimate $G(1 - u_t^0)$, let $L\subseteq W$ be a finite
point grid. Set
\begin{equation}
\label{e:Fhat}
\widehat{G(1-u_t^0)} := \frac{
\sum_{l_k \in L \cap W_{\ominus t} } \prod_{x\in X \cap B(l_k, t)} 
\left[ 1 - \frac{\bar \lambda}{\lambda(x)} \right]
}{
\# L \cap W_{\ominus t}
},
\end{equation}
where $W_{\ominus t} $ is the eroded set 
$\{ x\in W: d(x, \partial W) \geq t \} = 
\{ x\in W: x + B(0,t) \subseteq W \}$.
Note that $\widehat{G(1-u_t^0)}$ is an estimator as for all grid points
$l_k \in W_{\ominus t}$ the ball $B(l_k,t)$ is fully contained in $W$ so
that no points of $X\setminus W$ are needed for the computation of the
product in the numerator of (\ref{e:Fhat}).

Similarly, set
\begin{equation}
\label{e:Ghat}
\widehat{G^{!a}(1-u_t^a)}  =  
\frac{
\sum_{x_k \in X\cap W_{\ominus t}} \prod_{x \in X\setminus \{ x_k \} \cap B(x_k, t)} 
\left[ 1 - 
 \frac{\bar\lambda}{\lambda(x)}
\right] }
{ \# X\cap W_{\ominus t}} .
\end{equation}
Compared to (\ref{e:Fhat}), the grid points $l_k$ are replaced by the points 
$x_k$ of $X \cap W_{\ominus t}$. Again, (\ref{e:Ghat}) is a function of
$X\cap W$ only.

\begin{prop}
Under the assumptions of Theorem~\ref{t:gfl}, the estimator
(\ref{e:Fhat}) is unbiased, (\ref{e:Ghat}) is ratio-unbiased.
\end{prop}

\begin{proof}
We claim that 
\begin{equation}
\label{c:F}
\EE\left[
\prod_{x\in X \cap B(l_k, t)} 
\left( 1 - \frac{\bar \lambda}{\lambda(x)} \right) \right] = 
G(1 - u_t^0)
\end{equation}
for all $l_k \in L\cap W_{\ominus t}$. To see this, note that 
\begin{equation}
\label{e:prod}
\prod_{x\in X \cap B(l_k, t)} 
\left( 1 - \frac{\bar \lambda}{\lambda(x)} \right) 
=
1 + \sum_{n=1}^\infty \frac{(-\bar\lambda)^n}{n!} {\sum}^{\neq}_{x_1,
\dots, x_n \in X\cap W} \prod_{i=1}^n \frac{ 1\{ x_i - l_k \in B(0,t) \}}{
\lambda(x_i)}.
\end{equation}
Hence, under the assumptions made, the expectation of (\ref{c:F}) can be
expressed as
\[
1 + \sum_{n=1}^\infty \frac{(-\bar\lambda)^n}{n!} \int_{B(l_k,t)} \cdots
\int_{B(l_k,t)} \frac{ \rho^{(n)}(x_1, \dots, x_n )}{
\prod_{i=1}^n \lambda(x_i )} \, dx_1 \cdots dx_n,
\]
which, because of the translation invariance of the integrands, 
reduces to
\[
1 + \sum_{n=1}^\infty \frac{(-\bar\lambda)^n}{n!} \int_{B(0,t)} \cdots
\int_{B(0,t)} \frac{ \rho^{(n)}(x_1, \dots, x_n)}{
\prod_{i=1}^n \lambda(x_i)} \, dx_1 \cdots dx_n = G(1 - u_t^0).
\]
This proves the claim, from which the  unbiasedness of (\ref{e:Fhat}) follows.

Next, turn to the numerator of (\ref{e:Ghat}). By the definition of Palm
distributions and the reduced Campbell--Mecke theorem \cite[p.\107]{Stoyetal87},
its expectation can be expressed as
\[
\int \int_{W\ominus t} \lambda(x) \, \prod_{y\in \varphi} \left(1 - \frac{\bar\lambda \,
1\{ y \in B(x, t)\}}{\lambda(y)} \right) d\PP^{!x}(\varphi) \, dx.
\]
By (\ref{e:prod}), Fubini, the first equation in the proof of Theorem~\ref{t:gfl},
and (\ref{e:Gpalm}), the Palm expectation in the integrand equals 
$G^{!0}( 1 - u_t^0 )$ for almost all $x$, hence the numerator of (\ref{e:Ghat})
equals 
\(
G^{!0}( 1 - u_t^0 ) \, \int_{W\ominus t} \lambda(x) \, dx.
\)
As the expectation of the denominator is equal to $\int_{W\ominus t} \lambda(x) \, dx$,
(\ref{e:Ghat}) is ratio-unbiased as claimed.
\end{proof}
\begin{figure}[bhtp]
\begin{center}
\centerline{
\epsfxsize=0.25\hsize
\epsfysize=0.25\hsize
\epsffile{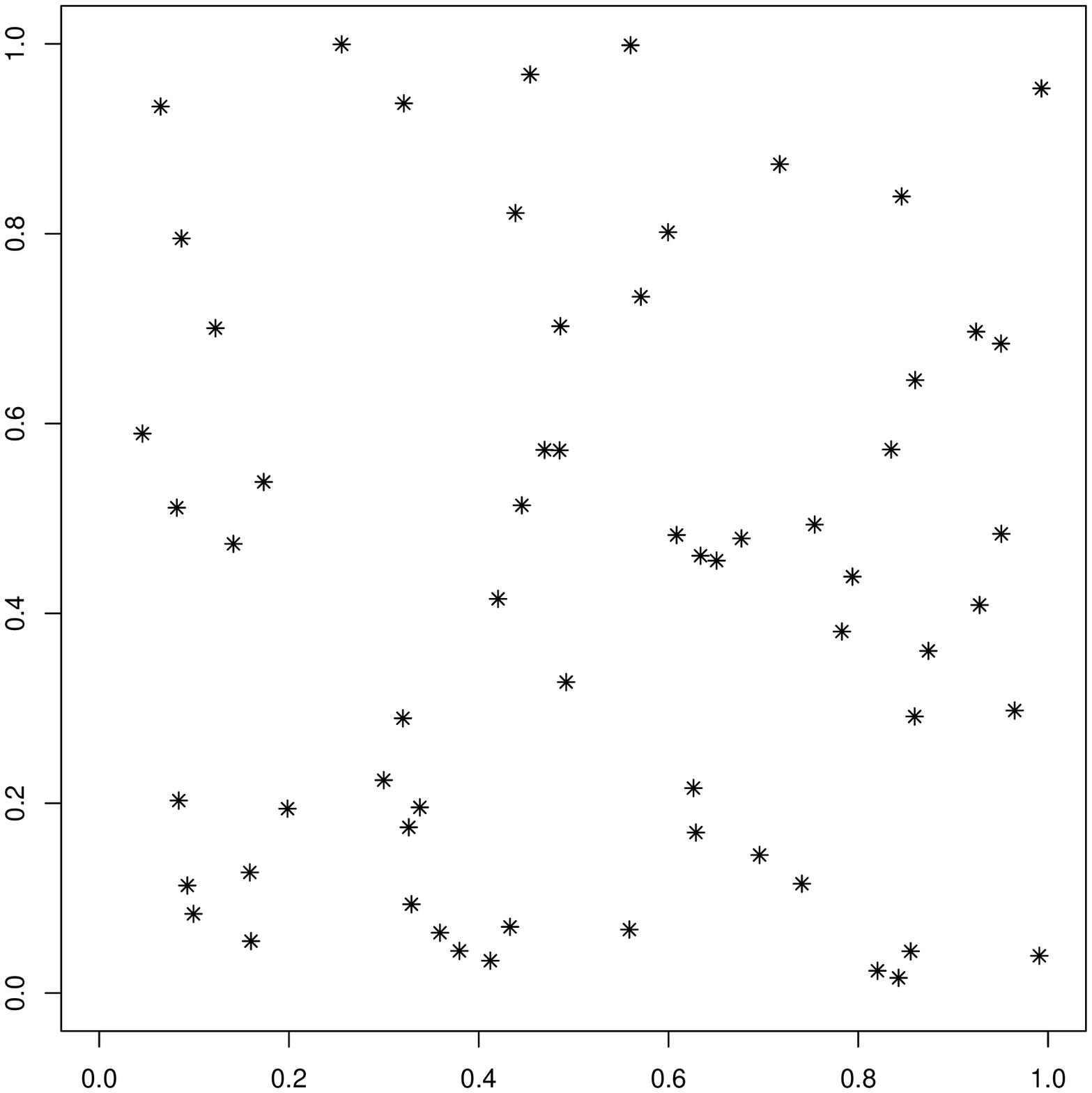}
\epsfxsize=0.25\hsize
\epsfysize=0.25\hsize
\epsffile{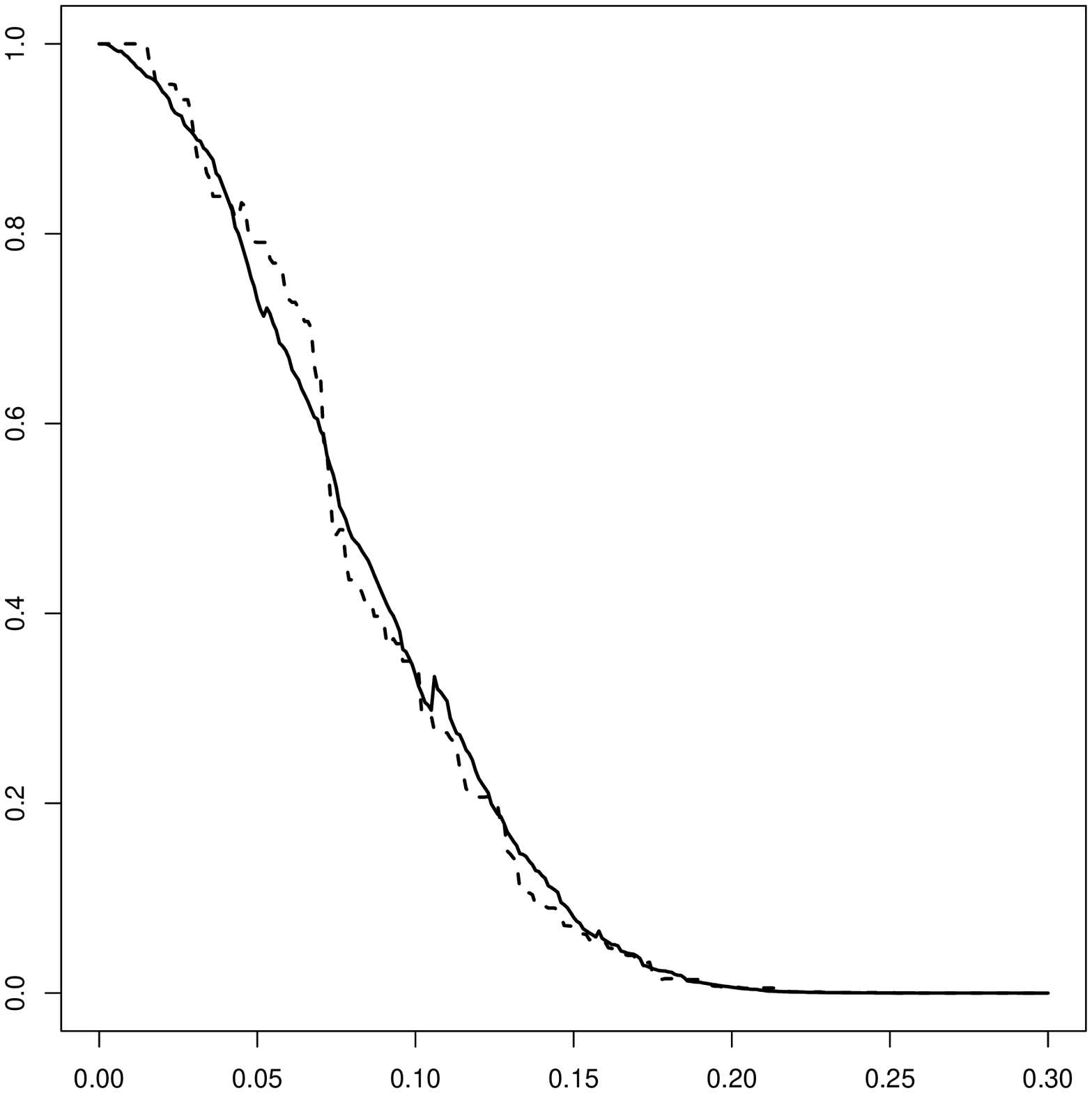}
\epsfxsize=0.25\hsize
\epsfysize=0.25\hsize
\epsffile{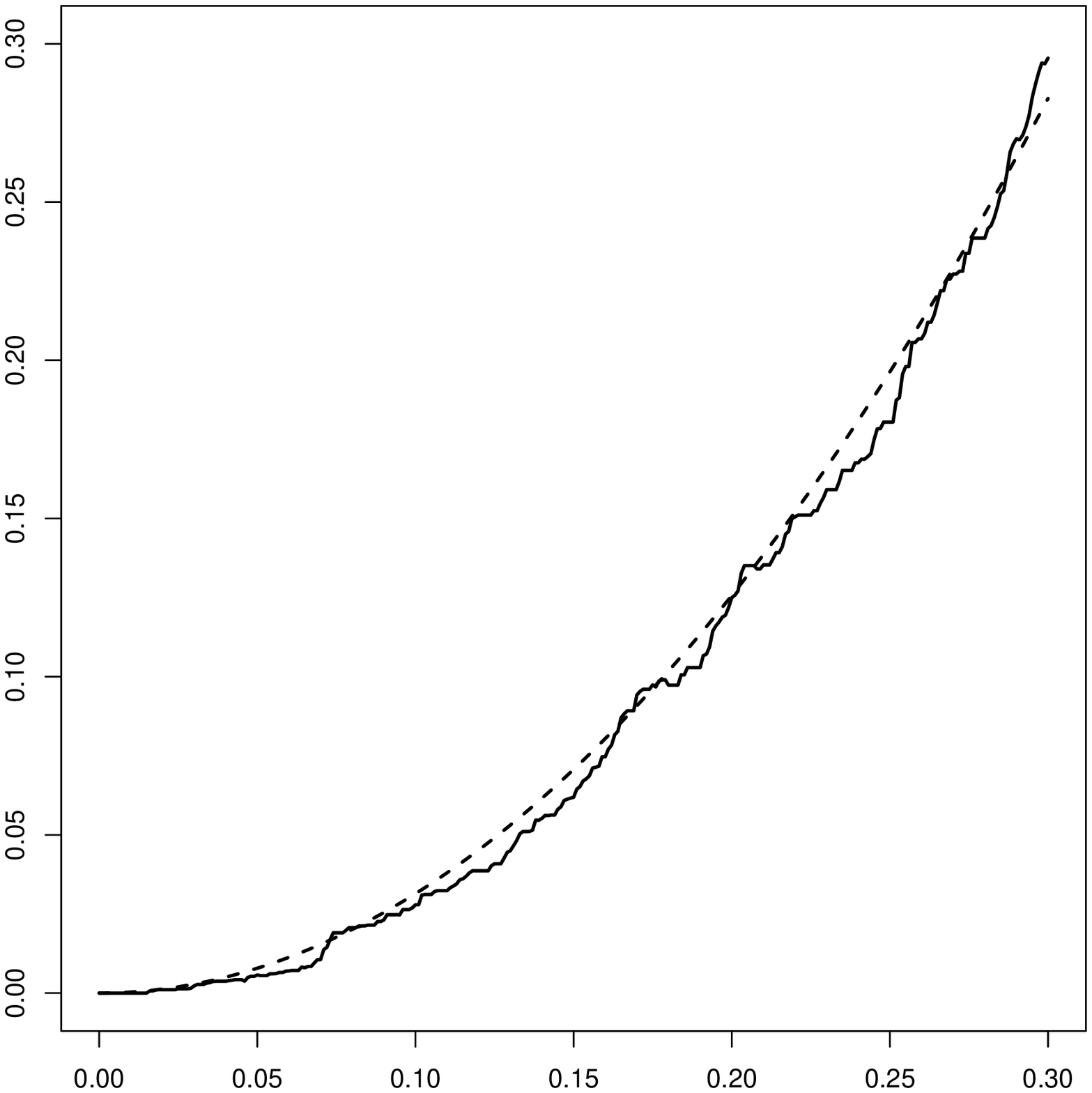}
}
\centerline{
\epsfxsize=0.25\hsize
\epsfysize=0.25\hsize
\epsffile{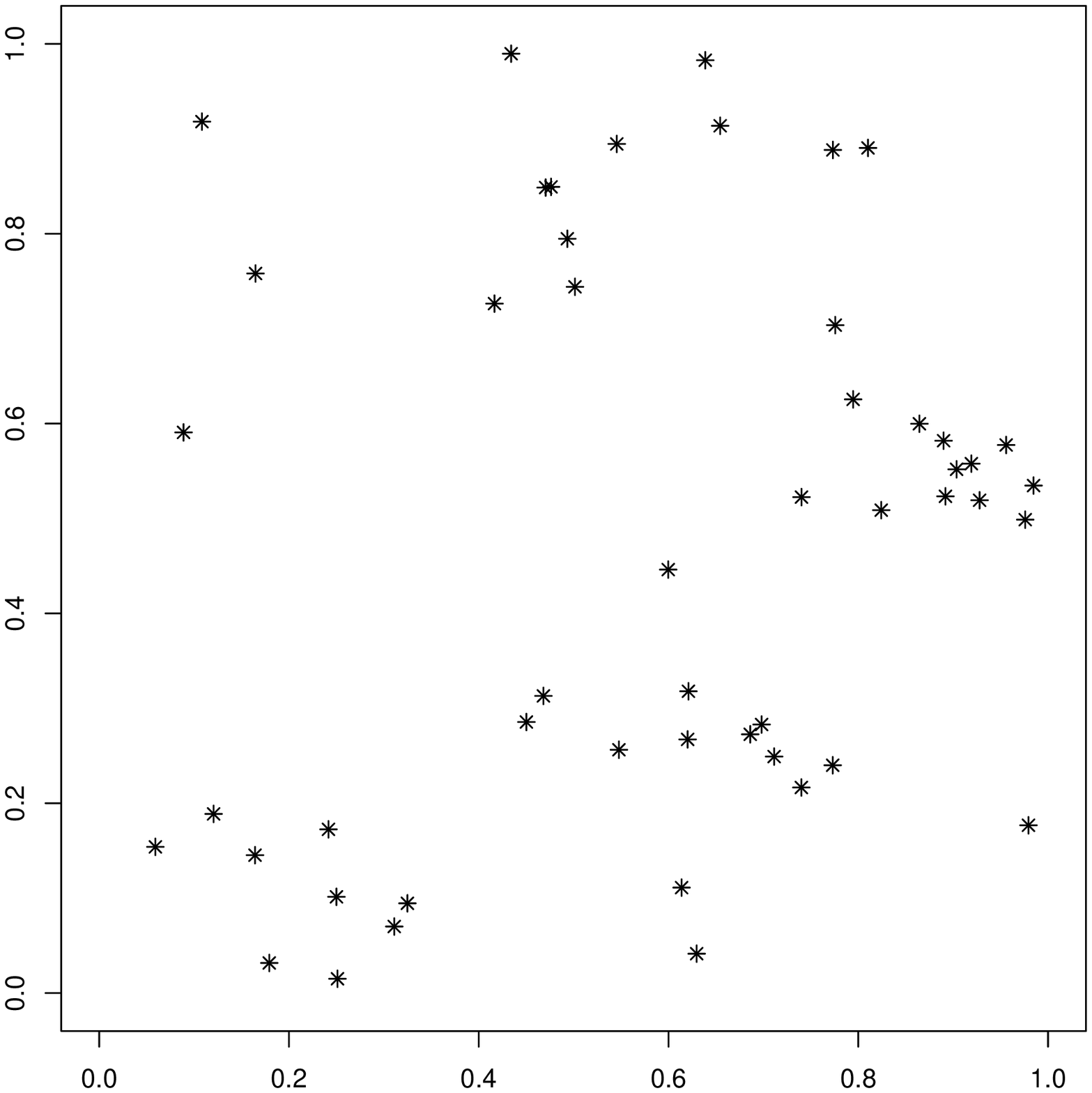}
\epsfxsize=0.25\hsize
\epsfysize=0.25\hsize
\epsffile{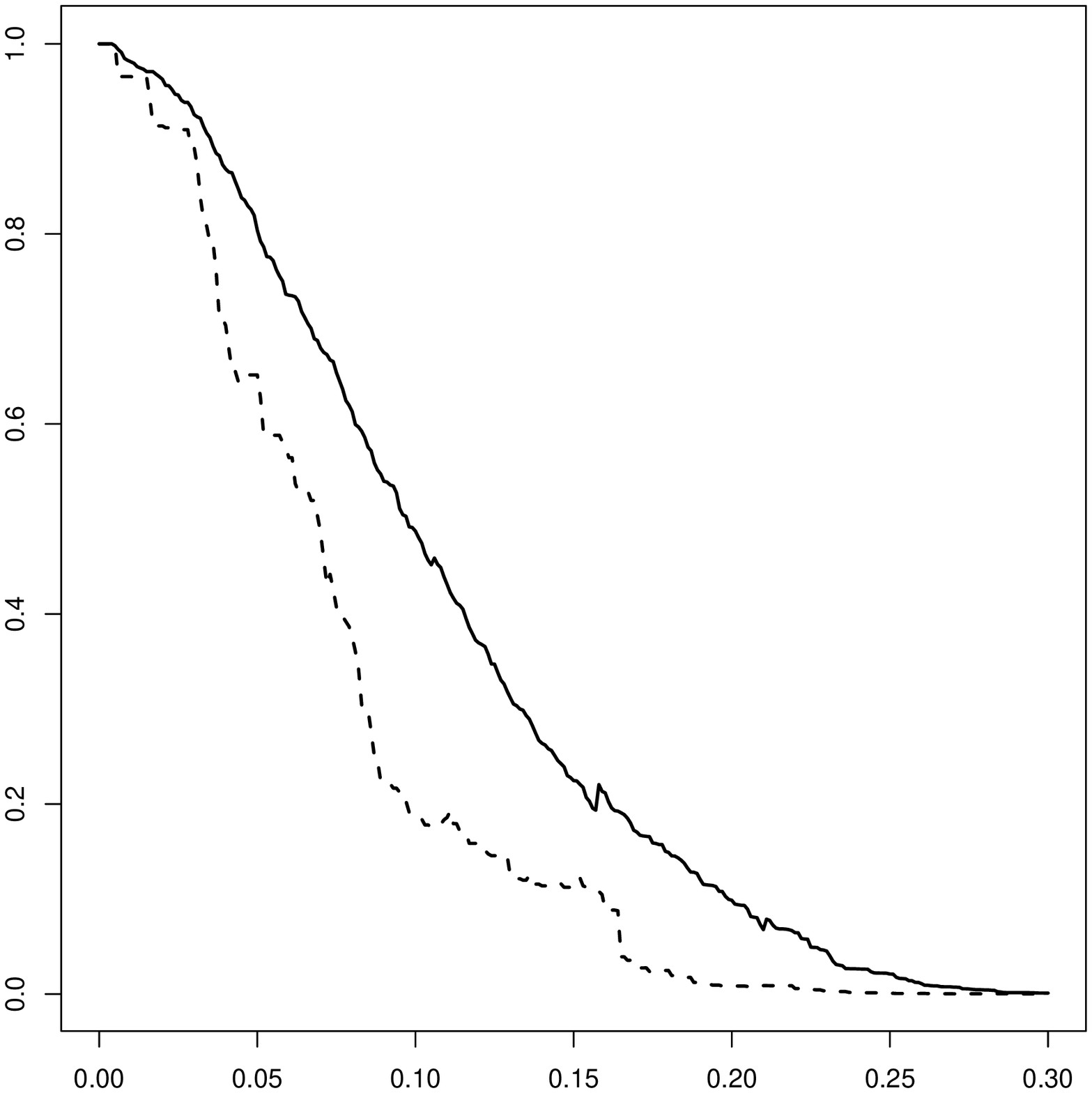}
\epsfxsize=0.25\hsize
\epsfysize=0.25\hsize
\epsffile{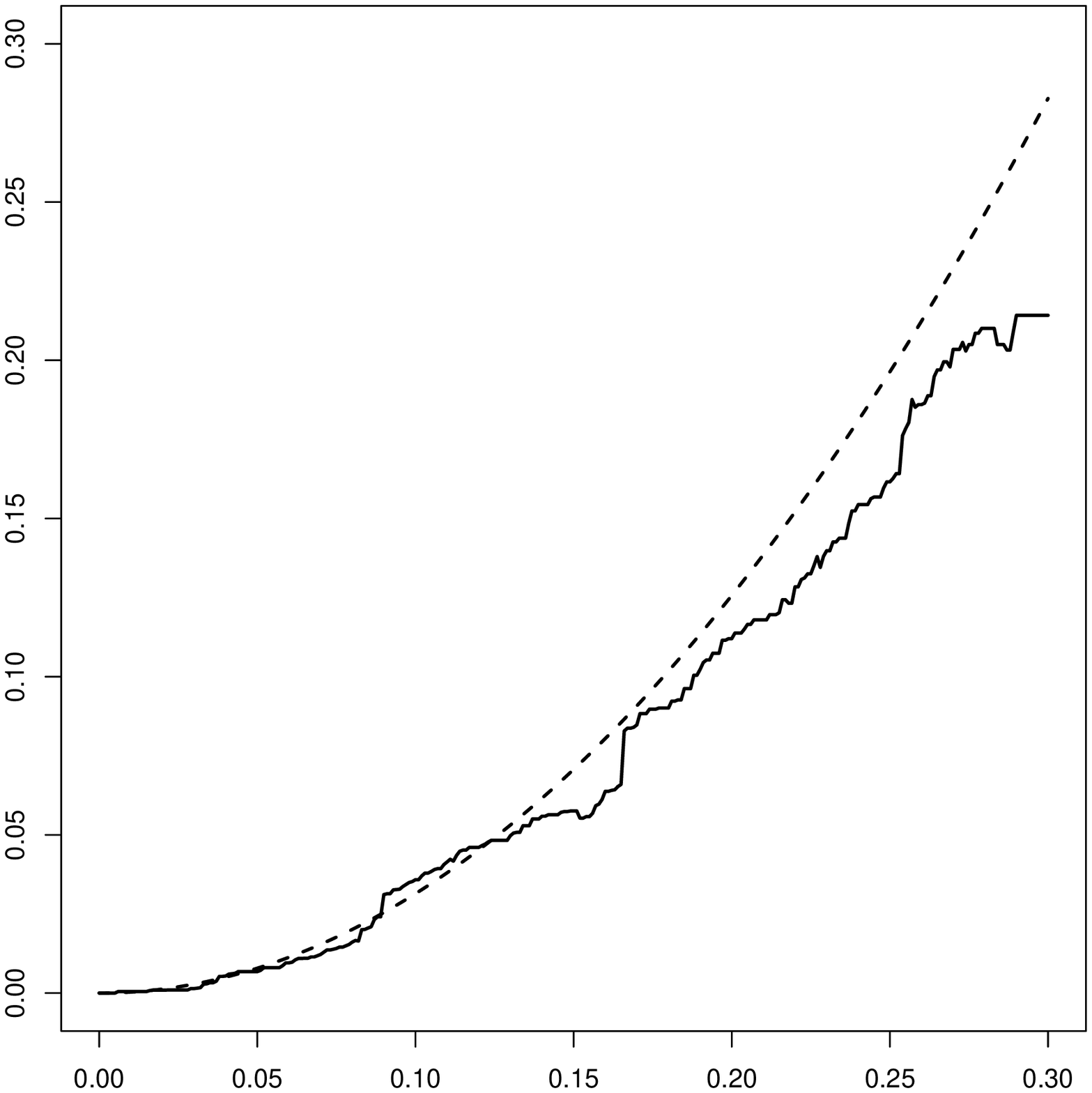}
}
\centerline{
\epsfxsize=0.25\hsize
\epsfysize=0.25\hsize
\epsffile{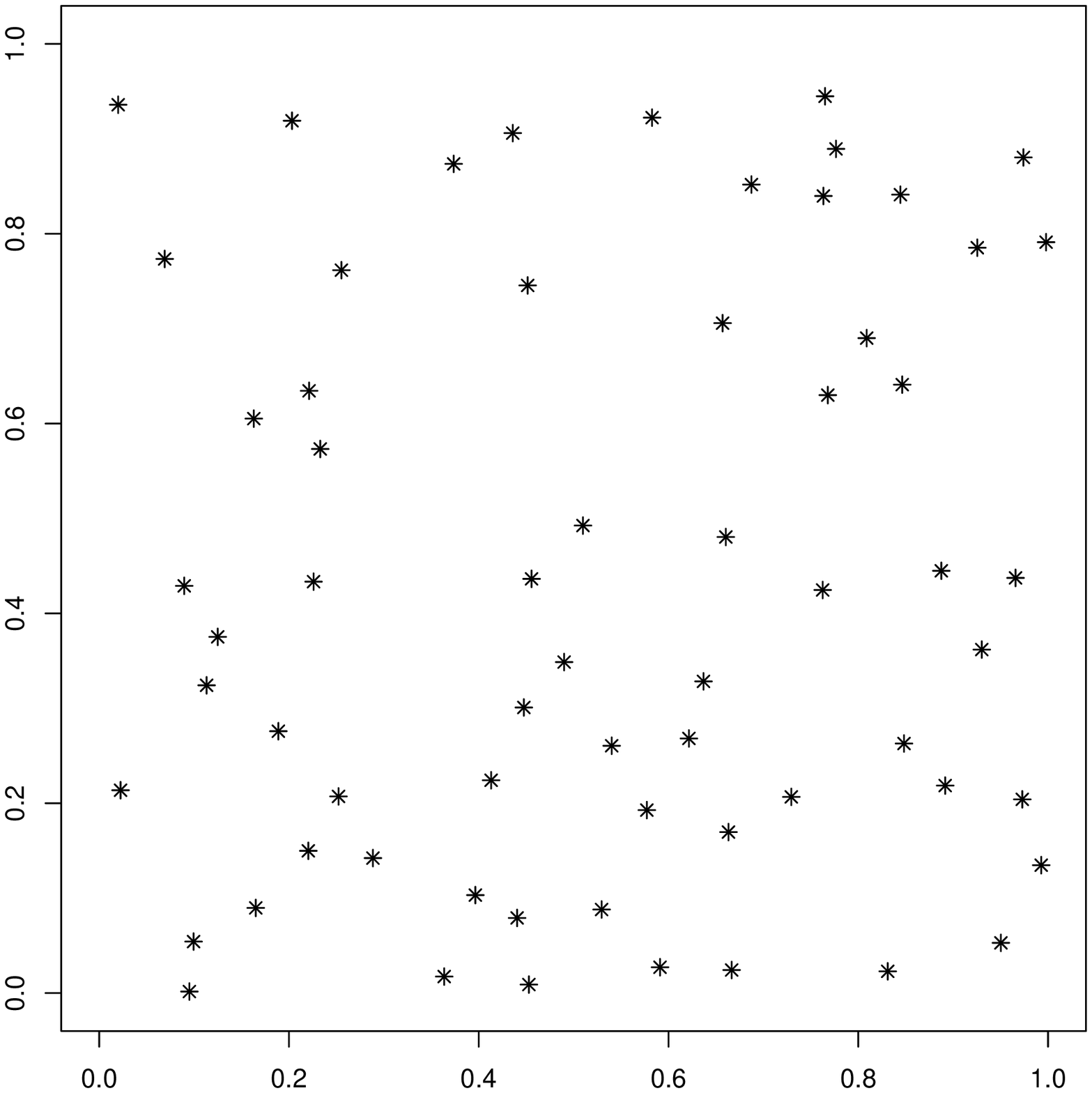}
\epsfxsize=0.25\hsize
\epsfysize=0.25\hsize
\epsffile{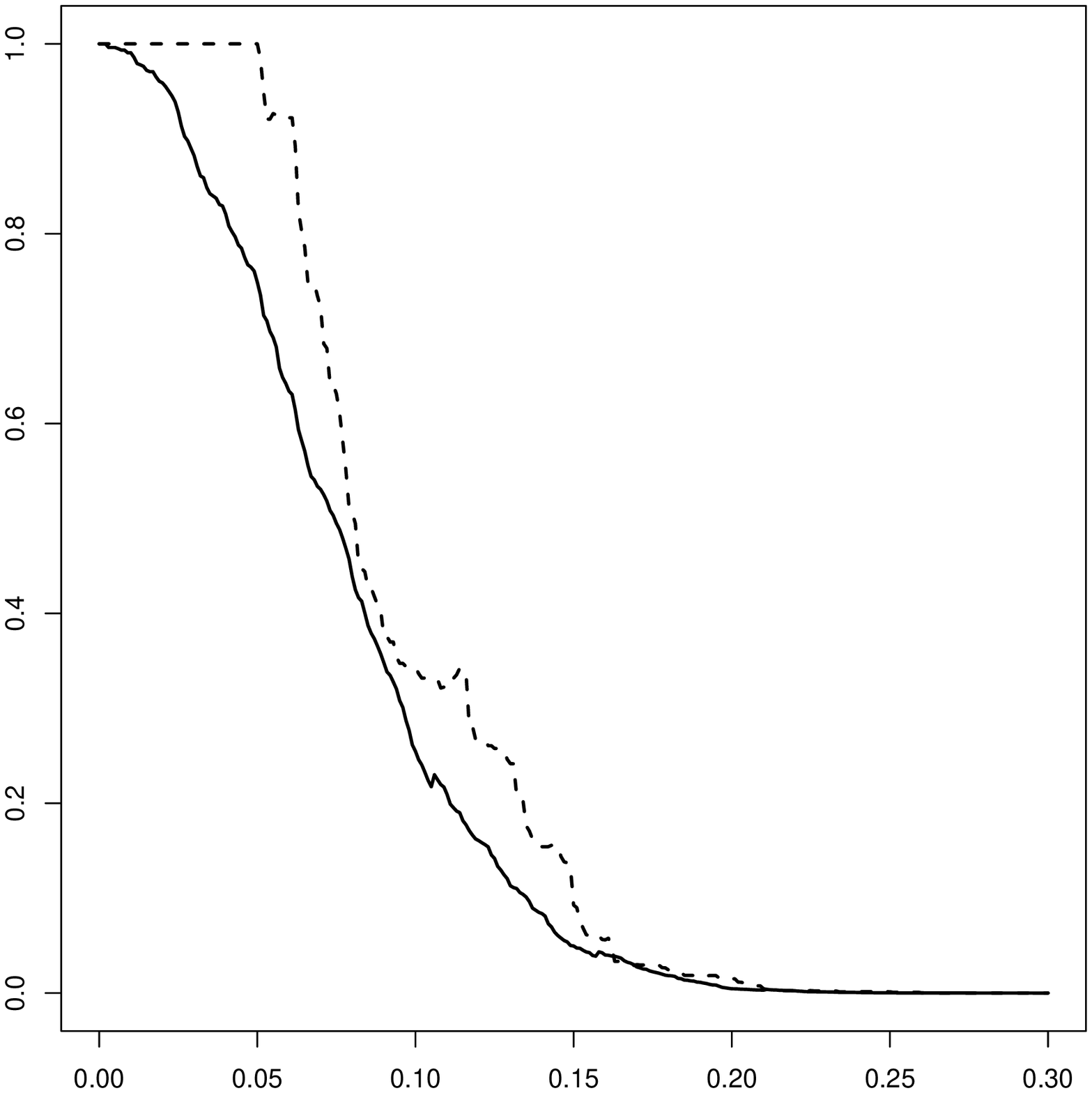}
\epsfxsize=0.25\hsize
\epsfysize=0.25\hsize
\epsffile{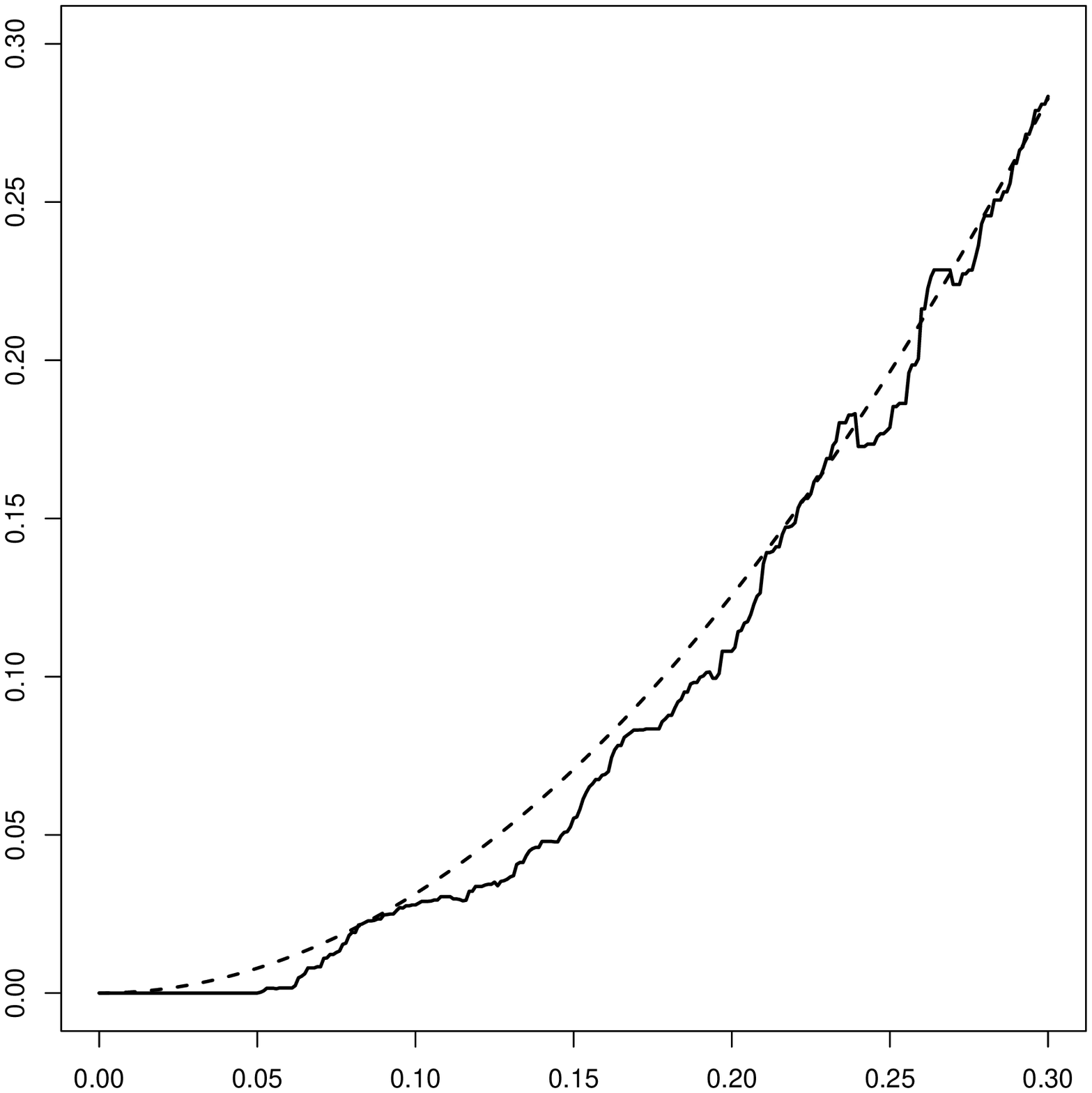}
}
\end{center}
\caption{Each row contains a realisations of a point process in the leftmost
frame, the graphs of (\ref{e:Fhat}) (solid line) and (\ref{e:Ghat}) (dashed line)
in the middle frame, and the graph of $\widehat{K_{\rm{inhom}}(t)}$ (solid line)
compared to $\pi \, t^2$ (dashed line) in the rightmost frame. The models are a 
Poisson point process (top row), a log Gaussian Cox process (middle row), and 
a thinned hard core process (bottom row).}
\label{F:Poisson}
\end{figure}

\section{Examples}
\label{S:examples}

In order to see how $J_{\rm{inhom}}(t)$ works in practice, we simulated
realisations of three of the models presented in Section~\ref{S:theo}. Typical 
patterns are displayed in the leftmost column of Figure~\ref{F:Poisson}.
In all three images a smooth intensity gradient can be observed: more points 
are located near the bottom of the square than near the top. However, the 
interaction structure seems different. For example, the middle picture contains 
groups of points that are close together, with large gaps in between the clusters.  
In the lower picture on the other hand, points seem to avoid being very close 
together and are more evenly spaced out. In the top picture, both very small 
and very large interpoint distances occur. 
In order to quantify the above qualitative remarks, we applied the ideas
presented in this paper and compared the results to those obtained by
a second order analysis. To simulate the patterns and calculate the 
estimators, the {\tt{R}} packages {\tt{spatstat}}\footnote{Adrian 
Baddeley, email: adrian@maths.uwa.edu.au; and Rolf Turner, email: 
r.turner@auckland.ac.nz} and {\tt{RandomFields}}\footnote{Martin 
Schlather, email: martin.schlather@ math.uni-goettingen.de} were used.

\paragraph{Poisson point process}
The first example is a heterogeneous Poisson point process with intensity
function
\(
\lambda(x,y) = 100 \, e^{-y}.
\)
Note that the mean number of points is $100 ( 1 - e^{-1}) \approx 
60$ per unit area. A realisation is shown in the top left frame in
Figure~\ref{F:Poisson}.  The top middle frame shows (\ref{e:Fhat})
(solid line) and (\ref{e:Ghat}) (dashed line). It can be seen that the graphs
lie close together, in accordance with the fact that for any Poisson
point process, $J_{\rm{inhom}} \equiv 1$. For comparison, 
the plug in minus sampling estimator of $K_{\rm{inhom}}$ is shown as the
solid line in the top right frame. Again, the graph is close to that of the
theoretical value $\pi \, t^2$ (dashed line in the top right frame). 

\paragraph{Log Gaussian Cox process}
The second example is a log Gaussian Cox process. The defining Gaussian
random field has exponentially decaying correlation function, unit variance, 
and mean function $\mu$ satisfying
\(
e^{ \mu(x,y) } = 100 \, e^{ -y - {1}/{2} }.
\)
Note that the intensity function of the Cox process thus defined coincides 
with that of the Poisson point process discussed above.
A realisation is shown in the middle row's leftmost frame in 
Figure~\ref{F:Poisson}.  The middle frame in the same row show (\ref{e:Fhat})
(solid line) and (\ref{e:Ghat}) (dashed line). Note that the graph of (\ref{e:Ghat}) 
lies well below that of (\ref{e:Fhat}), indicative of attraction between points due
to the positive correlation of $Z$ after accounting for the inhomogeneity.
For comparison, the plug in minus sampling estimator of $K_{\rm{inhom}}$ is
shown as the solid line in the rightmost frame in the middle row. From about 
$t=0.13$, the estimated value is smaller than $\pi \, t^2$.

\paragraph{Thinned hard core process}
The third example is a thinned hard core (Strauss) process defined by its
conditional intensity $\beta \, 1\{ d(x, X\setminus \{x\} >  R \}$. 
A realisation for $\beta = 200$, $R=0.05$ and retention probability 
\(
p(x,y) = e^{-y}
\)
is shown in the bottom left frame in Figure~\ref{F:Poisson}.  
The middle frame in the bottom row show (\ref{e:Fhat}) (solid line) and 
(\ref{e:Ghat}) (dashed line). Note that the hard core distance is clearly
reflected in the flat initial segment in the graph of (\ref{e:Ghat}), which 
lies abovethe graph of (\ref{e:Fhat}) up to about $r=0.2$, indicative of the
inhibition between points due to that present in the underlying hard core
process after accounting for the inhomogeneity. For comparison, the plug in
minus sampling estimator of $K_{\rm{inhom}}$ is shown as the solid line in the
bottom right frame.  The estimated value is smaller than that of a 
Poisson point process up to about $t=0.2$ confirming the picture painted by the
$J_{\rm{inhom}}$-function approach.

\section{Summary and extensions}

In this paper, we defined a $J$-function for intensity-reweighted moment
stationary point processes in terms of their $n$-point correlation functions
and gave representations in terms of the generating functional and conditional
intensity. We calculated $J_{\rm{inhom}}$ explicitly for the three 
representative classes of intensity-reweighted moment stationary point
processes presented in \cite{BaddetalK00}, derived an estimator, and
presented simulation examples. 

Although this paper focussed on point processes on $\R^d$, the approach 
may be extended to space time or marked point processes. First,
assume that $Y$ is a simple point process on the product space $\R^d\times \R$ 
equipped with the supremum distance whose intensity function $\lambda(\cdot)$ 
exists and $\inf_{(x,t)} \bar\lambda(x,t) > 0$. Furthermore assume all order 
factorial moment measures exist as locally finite measures that have Radon--Nikodym 
derivatives $\rho^{(n)}$ with respect to the $n$-fold product measure of 
$\ell$ with itself, $n\in \N$, and the corresponding $n$-point correlation
functions are translation invariant in both components. Define $J_n$ as in 
Definition~\ref{d:J}, from which an inhomogeneous space time version of the 
$J$-function can be defined. If the series is truncated at $n=1$, one obtains
\[
J^{ST}_{\rm{inhom}}(t) - 1 \approx -\bar\lambda \int_{-t}^t \int_{||x||\leq t} 
\xi_2((0,0), (x,s)) \, dx \, ds ,
\]
which corresponds to the $K^*_{ST}$-approach of Gabriel and Diggle
\cite{GabrDigg09}.  If space and time are scaled differently, see
Section~\ref{S:scaling}, $J^{ST}_{\rm{inhom}}(t,s)$ becomes a function of
two variables, one for spatial distances, the other for time differences,
which is more natural in many applications.

For marked point processes on $\R^d$ with marks in some Polish space
$M$ equipped with a finite reference measure $\nu$, make the same 
assumptions as above for space time point processes except that the $n$-point
correlation functions are required to be translation invariant in the spatial
component only.
For any Borel set $B\subseteq M$ and $n\in\N$, set $J_n^B(t)$ equal to the 
common value of
\[
 \frac{1}{\nu(B)} \int_B\int_{B(0,t)\times M} \cdots \int_{B(0,t)\times M} 
 \xi_{n+1}((a,b), y_1 + a, \dots, y_n + a) \, d\nu(b) \, d\ell\times\nu(y_1)
 \cdots d\ell\times\nu(y_n)
\]
for almost all $a\in \R^d$ and define a family of inhomogeneous $J$-functions
with respect to the mark set $B$ as in Definition~\ref{d:J}. Under suitable
regularity conditions,
\[
J_{\rm{inhom}}^B(t) = \frac{G_B^{!0}(1 - u_t^0)}{G(1-u_t^0)},
\]
where $u_t^a(y=(x,m)) = \bar\lambda 1\{ x \in B(a,t) \} / \lambda(y)$ and
\[
G_B^{!x}( 1 - u_t^x) = \frac{1}{\nu(B)} \int_B \int \left[
  \prod_{y\in Y} (1-u_t^x(y)) \right] d\nu(b) d\PP^{!(x,b)}(Y),
\]
which can be estimated using minus sampling ideas.


\end{document}